\documentclass[10pt]{ijnam}
\usepackage{graphicx,subfigure,array}
\usepackage{hyperref,hypernat,mathrsfs}

\def\currentvolume{6}
\def\currentissue{4}
\def\currentyear{2009}
\def\month{December}
\copyrightinfo{2009}{} % copyright year

\numberwithin{equation}{section}

\newtheorem{theorem}{Theorem}[section]
\newtheorem{lemma}[theorem]{Lemma}
\newtheorem{proposition}[theorem]{Proposition}
\newtheorem{remark}[theorem]{Remark}

\newtheorem{definition}[theorem]{Definition}
\newtheorem{example}[theorem]{Example}
\newtheorem{assumption}[theorem]{Assumption}

\newcommand{\secref}[1]{\S\ref{#1}}
\newcommand{\lemref}[1]{Lemma~\ref{#1}}
\newcommand{\propref}[1]{Proposition~\ref{#1}}

\newcommand{\assref}[1]{Assumption~\ref{#1}}

\newcommand{\exref}[1]{Example~\ref{#1}}
\newcommand{\figref}[1]{Figure~\ref{#1}}

\newcommand{\tabref}[1]{Table~\ref{#1}}

\newcommand{\dd}{\textrm{d}}
\newcommand{\ii}{\textrm{i}}
\newcommand{\nn}{\nonumber}

\newcommand{\be}{\begin{eqnarray}}
\newcommand{\ee}{\end{eqnarray}}
\newcommand{\bal}{\begin{aligned}}
\newcommand{\eal}{\end{aligned}}
\newcommand{\bes}{\begin{eqnarray*}}
\newcommand{\ees}{\end{eqnarray*}}
\newcommand{\bs}{\begin{subeqnarray}}
\newcommand{\es}{\end{subeqnarray}}
\newcommand{\bss}{\begin{subeqnarray*}}
\newcommand{\ess}{\end{subeqnarray*}}

\newcounter{saveeqn}
\def\hat{\widehat}
\def\supp{\operatorname{supp}}

\def\hu{\widehat u}
\def\hv{\widehat v}

\def\BbbR{\mathbb R}

\def\tt{T}

\def\p{\partial}

\def\O{\Omega}
\def\V{\mathcal V}
\def\Z{\mathcal Z}
\def\om{\omega}

\def\al{\alpha}

\def\ga{\gamma}
\def\Ga{\Gamma}
\def\G{\Gamma}

\def\Re{\operatorname{Re}}
\def\Im{\operatorname{Im}}
\def\and{\quad\text{and}\quad}

\def\hu{\widehat u}

\def\si{\sigma}
\def\rd{\partial}
\def\esssup{\operatorname{ess.\, sup}}
\def\<{\left\langle}
\def\>{\right\rangle}

\newcolumntype{x}[1]{>{\centering\hspace{0pt}}p{#1}} %>{\raggedleft\hspace{0pt}}p{#1}}%right

\begin{document}
%%%%%%%%%%%%%%%%

\title[Laplace transformation method for the Black-Scholes equation]{Laplace transformation method for the Black-Scholes equation}

% author and address
\author[H. Lee]{Hyoseop Lee}
\address{Interdisciplinary Program in Computational Science \& Technology,
         Seoul National University,
         Seoul 151--747,
         Korea
}
\email{hyoseop2@snu.ac.kr}
\urladdr{http://www.nasc.snu.ac.kr/hslee/}

% two authors have the same address
\author[D. Sheen]{Dongwoo Sheen}
\address{
Department of Mathematics, and
Interdisciplinary Program in Computational Science \& Technology,
         Seoul National University,
         Seoul 151--747,
         Korea
}
\email{sheen@snu.ac.kr}
\urladdr{http://www.nasc.snu.ac.kr/sheen/}

% dedication
%\dedicatory{This paper is dedicated to our authors}

% communicated
\commby{Yanping Lin}

% Received by the editors ?
%\date{January 25, 2009}% and, in revised form, March 22, 2004.}

% it is suggested to put it in Acknowledgments
\thanks{The research of HL was partially supported by  KRF-2007-C00001
and that of DS by
KRF-2007-C00031. To appear in  \href{http://www.math.ualberta.ca/ijnam/}{\underline{International Journal of Numerical Analysis \& Modeling.}}
}

\subjclass[2000]{91B02, 44A10, 35K50}

\abstract{In this paper we apply the innovative Laplace transformation method
introduced by Sheen, Sloan, and Thom\'ee (IMA J. Numer. Anal., 2003) to solve
the Black-Scholes equation. The algorithm is of arbitrary high
convergence rate and  naturally parallelizable.
It is shown that the method is very efficient  for 
calculating various option prices.
Existence and uniqueness properties of the Laplace transformed Black-Scholes equation
are analyzed.
Also a transparent boundary condition associated with the Laplace
transformation method is proposed.
Several numerical results  for various options under various situations
confirm the efficiency, convergence and parallelization property of the proposed scheme.
}

\keywords{Black-Scholes equation, basket option, Laplace inversion, 
parallel method, transparent boundary condition}

\maketitle
\section{Introduction}
As stock markets have become more sophisticated, so have
their products. The simple buy/sell trades of the early markets have been replaced
by more complex financial options and derivatives.
These contracts can give investors various opportunities 
to tailor their dealings to their
investment needs.

One of the main concerns about financial options is what the exact values of
options are. For the simplest model in the case of constant coefficients,
an exact pricing formula was
derived by Black and Scholes, known as the Black-Scholes formula.
However, in the general case of time and space dependent coefficients 
the exact pricing formula are not yet established, and thus numerical
solutions have been used.

In order to describe an option price, let $x, K, t$ and $T$ denote the underlying asset price,
the strike price, the time to maturity, and the expiry date of an option,
respectively. As usual, $\si$ and $r$ represent the volatility of the underlying asset
and the risk-free interest rate of the market, respectively.
In this paper, we assume that $\si$ and $r$ depend on $x$ only.
Then a European option price $u(x,t)$ satisfies
the Black-Scholes equation:
\be
\frac{\rd u}{\rd t} -\frac12{\si^2x^2}\frac{\rd^2u}{\rd x^2}
-rx\frac{\rd u}{\rd x} + ru = 0, \quad (x,t)\in(0,\infty)\times(0,T],
\label{eq:bs-one}
\ee
where an initial condition $u_0(x)=u(x,0)$
is given by the initial contract of an option.
The basket option based on $n$ assets $\mathbf{x}=(x_1, \dots, x_n)$ satisfies
\be \label{eq:bs-basket}
&&  \frac{\rd u}{\rd t}
  -\frac12 \sum^{n}_{i,j=1}a_{ij}x_ix_j\frac{\rd^2 u}{\rd x_i \rd x_j}
  - \sum^n_{i=1}rx_i\frac{\rd u}{\rd x_i} + ru=0, \\
&&\qquad \qquad \qquad (\mathbf{x},t)\in(0,\infty)^n\times(0,T], \nonumber
\ee
where $a_{ij}=\sum^{n}_{k=1}\si_{ik}\si_{jk}$, with $\si_{ij}$ representing the corelation between the assets $x_i$ and $x_j$.

Several numerical methods have been used for solving the Black-Scholes equation,
for example in \cite{fin-fdm, tool-fin} and \cite{monte-fin} and the
references therein, one can find popular numerical schemes for option
pricing.
Usually the time marching methods such as forward Euler,
backward Euler and Crank-Nicolson schemes are used with a suitable
spatial discretization scheme.
In spite of the popularity of these time marching methods,
a critical drawback of these schemes is that they usually require as many
time steps as spatial meshes to balance the errors arising from discretization.
In particular, for the estimation of basket options of reasonable size,
the usual time marching schemes seem to be too slow in practice
%too slow to be used in practice
since the cost of solving an elliptic system to advance to a next time step
is usually expensive. It is thus highly desirable to solve as small a number of
%highly demanding to solve as small number of
elliptic
solution steps as possible as well as to apply a very fast elliptic solver.

In this paper, we will focus on minimizing the number of elliptic solution steps
by proposing the 
Laplace transformation method for the Black-Scholes equation, which is also
naturally parallelizable.
It will be shown that our method
can dramatically reduce the computing time compared to the time marching schemes.
%time marching ones. 
Suitable contours should be chosen in order to have very fast
convergence, and for this, we will estimate the resolvent of the Black-Scholes
equation. Also, an exact transparent boundary condition will be given at which
the infinite spatial domain is truncated.

There have been some related works in which the Laplace transformation method
has been used, for instance in \cite{asian-lap-inv,
  american-lap, double-lap}.
However, in these earlier papers the Laplace transformation method has been used
to obtain the analytic solution of various options
rather than to develop an efficient numerical scheme.
%develop into efficient numerical schemes.
In particular, in \cite{american-lap} the partial Laplace transformation is applied
for American option pricing, and in \cite{mellin-cruz-baez, mellin-fin} the Mellin transformation
which is similar to the Laplace transformation is used to evaluate the analytic solution
of an option. Related with Laplace transformation methods %there are other approached based 
there are other approaches based
on the so-called {$\mathcal H$}-matrix approach; for instance, see
\cite{gavril02,  gav-hackbusch-khoromskij-05}, and so on.
Also, high-dimensional parabolic problems can be 
solved using sparse grids \cite{MGriebel_1994a, leevntvaar-oosterlee-06,
  oosterlee-sparsegrid-jcam-08, reisinger-wittum-07}. Application
of our Laplace transformation method using sparse grids to option pricing 
will also be interesting. Other approaches in the fast time-stepping methods
can be found in \cite{schwab-fast-04, schwab-04, matache-schwab-wihler-05}.

In the following section, we will briefly describe the Laplace transformation
method proposed by Sheen, Sloan, and Thom\'ee in \cite{sst-para2}
with its numerical procedure and convergence. Then in 
\secref{sec:lap-bs} we will examine
the properties of the Laplace transformed Black-Scholes equation
including the solvability of the transformed equation,
transparent boundary condition and the resolvent.
Finally in \secref{sec:num}  we will present several numerical results for various options
and various situations with the parallelization property of the proposed scheme.

%%%%%%%%%%%%%%%%%%%%%%%%%%%%%%%%%%%%%%%%%%%%%%%%%%%%%%%%%%%%%%%%%%%%%%%%%%%%%
% Section 2. Laplace Transformation
%%%%%%%%%%%%%%%%%%%%%%%%%%%%%%%%%%%%%%%%%%%%%%%%%%%%%%%%%%%%%%%%%%%%%%%%%%%%%
\section{The Laplace transformation and its inversion}\label{sec:laplace}
We begin with the abstract setting of a parabolic type equation so that
the proposed scheme can be applicable to various problems. Consider
\be\label{eq:parabolic}
\frac{\rd u}{\rd t}  + Au = f, \quad t\in(0.T]; \quad
u(0)=u_0,
\ee
where $u_0$ is a given initial function and $A$ a spatial elliptic operator
with its eigenvalues being located in the right half plane. (We added the
source term $f(x,t)$, which is not present in \eqref{eq:bs-one} or in
\eqref{eq:bs-basket}, in order to describe our method in more general setting.)
For each $z$ in the complex plane,
recall that the standard Laplace transform in time of a
function $u(\cdot,t)$ is given by
\be
\hu(\cdot,z):=\mathcal{L}[u](z)=\int_0^{\infty}u(\cdot,
t)e^{-zt}\,dt.\nn
\ee
Then the Laplace transformation of \eqref{eq:parabolic} is thus given in the form
\be\label{eq:lap-parabolic}
z\hat u + A\hat u &=& u_0 + \hat f(\cdot,z), \quad z\in \Gamma,
\ee
from which the solution $\hat u(z)= \hat u(\cdot,z)$ is formally given
by
\be\label{eq:lap-sol}
\hat u(\cdot,z) = (zI+A)^{-1} (u_0(\cdot)+\hat f(\cdot,z)),
\ee
for each $z.$ We suppose that the real parts of singular points of $\hat f(z)$
are less than some positive number.

The {\it Laplace inversion formula} (\cite{bromwich}) is  given by
\be\label{eq:inv-lap}
u(\cdot,t)=\frac{1}{2\pi \ii}\int_{\Ga} \hu(\cdot,z)e^{zt} \,dz,
\ee
where the integral contour $\Ga$ is a straight
line parallel to the imaginary axis expressed as
%, and can be written explicitly like followings
\be\label{eq:Gamma}
\qquad\Ga := \{z\in\mathbb{C}:z(\om)=\alpha+\ii\om \,\textrm{where}\,\, \om\in\BbbR\,\,\textrm{increases from}\, -\infty \,\textrm{to}\, +\infty\}.
\ee
The constant $\alpha\in \BbbR$ in the contour is called the Laplace convergence
abscissa, and the value of $\alpha$ is required to be greater than the real
part of any singularity of $\hu(z)$.

Inserting the explicit form of $z\in\Gamma$ given by \eqref{eq:Gamma} into Equation \eqref{eq:inv-lap}, one has
\begin{align}\label{eq:inv-lap-ch}
u(x,t)=\frac{e^{\alpha t}}{\pi}\int_{0}^{\infty}
\big[ \textrm{Re}\{\hu(x,\al+\ii\om)\}\cos(\om t) - \textrm{Im}\{\hu(x,\al+\ii\om)\}\sin(\om t)\big]\,d\om.
\end{align}
Denoting by $\sum'$ the summation with its first and the last summands being halved,
an application of the composite trapezoidal rule to this integral leads to the direct method
\be
u(x,t)\approx\frac{e^{\al t}}{T}
{{\sum}'}_{k=1}^{N-1}\left[
\textrm{Re}\{\hu(x,\al+\frac{k\pi \ii}{T})\}\cos(\frac{k\pi t}{T})
- \textrm{Im}\{\hu(x,\al+\frac{k\pi \ii}{T})\}\sin(\frac{k\pi t}{T})\right], \nn
\ee
for some sufficiently large $N$ with the length of two mesh points $\pi/T$.
Although this scheme can be easily implemented,
its convergence rate is slow due to the truncation and discretization errors.
In order to approximate the integration \eqref{eq:inv-lap-ch} fast and accurately,
there have been numerous modifications, such as \cite{cohen-laplace, crump,
  gavril02, gav-hackbusch-khoromskij-05,
  tal2,
  palencia04, sst-para, sst-para2,
jleesheen-laplace, jleesheen-back, mclean-thomee, mclean-sloan-thomee, tal, thomee-ijnam, weeks, weideman, widder}
and the references therein.
In this paper, we will use the deformation of the contour introduced in \cite{sst-para2},
which gives an arbitrary high-order convergence rate with a hyperbolic type deformation.

\subsection{Deformation of contour}

For a concrete mathematical analysis, we assume that
the spectrum $\sigma(A)$ of $A$ lies in a sector $\Sigma_\delta$ such that
\be
\si(A)\subset\Sigma_{\delta}=\{z\in\mathbb{C}: |\arg z| \le \delta,\,
z\ne 0,\,\delta\in(0,\frac{\pi}{2})\}, \nn
\ee
and the resolvent $(zI+A)^{-1}$ of $A$ satisfies
\be
\|(zI+A)^{-1}\|\le \frac{M}{1+|z|}, \quad \textrm{for }z\in\Sigma_{\delta}\cup B, \nn
\ee
where $B$ is a small circle at the origin.

The first restriction is required to avoid the singular points of
the integrand in \eqref{eq:inv-lap}.
%the Laplace transformed function of $u$.
Since the problem \eqref{eq:parabolic} has a solution of the form
\be\label{eq:inv-lap2}
u(t)\left(=u(\cdot,t)\right) = \frac1{2\pi\ii}
\int_\Ga \big(zI+A\big)^{-1}\big(u_0+\widehat{f}(z)\big)e^{zt}
\,dz\label{eq:para-inv},
\ee
{\it the integral contour has to be kept away from the spectrum
of $-A$ and the singular points of $\hat f(z),$}
 when we deform the contour.
In particular, since all eigenvalues
of $-A$ and the singularities of $\hat f(z)$ have real parts bounded by a
positive number, this restriction is natural.

Observe that if $z\in \Ga$ has negative real parts as
$|z|$ becomes large, the discretization error in numerically evaluating the
integrand in \eqref{eq:inv-lap2} will be reduced for positive $t$; thus it will be
desirable to deform the contour to the left half plane as long as all the
singularities are to the left of it. Based on this,
Sheen \textit{et al.} \cite{sst-para2} proposed the smooth contour
of hyperbola type as follows:
\be
\Ga =\{z\in\mathbb{C}:z(\om)=\zeta(\om)+\ii s\om, \quad\om\in\mathbb{R},
\quad \om\,\,\textrm{increasing}\}, \nn
\ee
where $\zeta(\om)=\ga-\sqrt{\om^2+\nu^2}$. In this case, since the contour
cuts the real line at $\gamma-\nu$, $\gamma$ and $\nu$ must be
selected such that $\gamma-\nu$ is larger than the negative of
the smallest eigenvalue of $A$ and the real parts of singularities of $\hat f(z).$
Also $s$ should be chosen such that all the singularities of $\hat u(\cdot,z)$
be to the left of the contour $\G.$

Using the above deformed contour, the inversion formula can be written as an infinite integral
with respect to a real variable,
\begin{align}
u(\cdot,t) = \frac{1}{2\pi\ii}\int^{\infty}_{-\infty}\hu(\cdot,\zeta(\om)+\ii s\om)(\zeta^{\prime}(\om)+\ii s)e^{(\zeta(\om)+\ii s\om)t}\dd\om. \nn
\end{align}
The infinite range of the above integration can be changed into to a finite region by
the change of variables of the form
\be
y(\om) = \tanh \big(\frac{\tau \om}{2}\big)\,\,\textrm{ and }\,\, \om(y) = \frac2{\tau}\tanh^{-1}(y) = \frac1{\tau}\log\frac{1+y}{1-y}\nn,
\ee
for some $\tau>0$ and $y\in(-1,1)$. The above change of variables reduces from
an integral on an infinite interval to one on a finite interval as follows:
\begin{align}\label{eq:int-finite}
u(\cdot,t) = \frac{1}{2\pi\ii}\int^{1}_{-1}\hu(\cdot,\zeta(\om(y))+\ii s\om(y))(\zeta^{\prime}(\om(y))+\ii s)e^{(\zeta(\om(y))+\ii s\om(y))t}\om^{\prime}(y)\dd y.
\end{align}

\subsection{Semi-discrete approximation}
The last integral formula \eqref{eq:int-finite} in the previous section can be
discretized in time using a quadrature rule.
Explicitly the semi-discrete approximation of $u(\cdot,t)$ is given by
\be
U_{N,\tau}(\cdot,t)=\frac{1}{2\pi\ii}\frac{1}{N}\sum^{N-1}_{j=-N+1}\hu(\cdot,z_j)\frac{\dd z}{\dd \om}(\om_j)\frac{\dd \om}{\dd y}(y_j)e^{z_jt},\label{eq:inv-lap-app}
\ee
where
\be
z_j=z(\om_j),\quad \om_j=\om(y_j)\quad\textrm{and}\quad y_j=\frac{j}{N},\quad\textrm{for}\,-N<j<N. \nn
\ee

It is proved in \cite{sst-para2} that the quadrature scheme
\eqref{eq:inv-lap-app}
is of arbitrary high-order spectral convergence rate if in particular the source term
$f$ has high-order regularity, stated as follows:
\begin{theorem}[Sheen-Sloan-Thom\'ee]\label{thm:time-order}
Let $u(t)$ be the solution of (\ref{eq:parabolic}) and
let $U_{N,\tau}(t)$ be its approximation
defined by (\ref{eq:inv-lap-app}). Assume that $\widehat{f}(z)$
is analytic to the right of the contour $\Ga$ and
continuous onto $\Ga$, with $\widehat{f}^{(j)}(z)$ bounded on
$\Ga$ for $j\le r$ and $r$ an integer $\ge 1$, Then,
for $t>r\tau$
\begin{equation}
\| U_{N,\tau}(t)-u(t)\|\le \frac{C_{r,s}}{N^r}
\Big( 1+t^r+\frac1{\tau^r}\Big)e^{\gamma t}
\Big( 1+\log_{+}\frac{1}{t-r\tau}  \Big)
(\|u_0\|+\max_{k\le r}\sup_{z\in \Gamma}
\|\widehat{f}^{(k)}(z)\|).
\end{equation}
\end{theorem}
Three important remarks should be stressed.
\begin{remark}
The implication of the above theorem without source term $f$ as in our
option pricing is that the scheme is of order $O(\frac1{N^r})$ with
an arbitrarily large $r>0$ since $\hat f\equiv 0$ is certainly analytic and
$\widehat{f}^{(r)}(z)$ is bounded on $\Ga$ for positive integer $r.$
This implies that the discretization errors in the time direction using the
Laplace transformation method will be negligible compared to those caused from
the spatial discretization part in solving the Black-Scholes equation.
\end{remark}
\begin{remark}
In the summand \eqref{eq:inv-lap-app}, an important observation is that
\[
\hu(\cdot,z_j)\frac{\dd z}{\dd \om}(\om_j)\frac{\dd \om}{\dd y}(y_j),\,\,j
=0,\cdots, N,
\]
are independent of $t$. Therefore, we only have to approximate
$\hu(\cdot,z_j)$ only once by solving the complex-valued elliptic problem
\eqref{eq:lap-parabolic}
for a set of $z_j, j =0, 1, \cdots, N.$ Then, if we need the option pricing
at a different time $t$, the same set of spatial solutions $\hu(\cdot,z_j), j
=0,1,\cdots, N$, can be used in the evaluation of the summation
\eqref{eq:inv-lap-app}
with the only change in $e^{z_jt},$ for the needed time $t.$
\end{remark}
\begin{remark}
Notice that each elliptic problem \eqref{eq:lap-parabolic} for a $z_j$
from the set of $z_j, j =0, 1, \cdots, N,$
is independent of other elliptic problems for the remaining $z_j$'s. This
will minimize communication times in solving the elliptic problems
\eqref{eq:lap-parabolic} in parallel by assigning each processor to solve
an independent elliptic problem without communicating with other processors
during solving its assigned problem.
\end{remark}
%%%%%%%%%%%%%%%%%%%%%%%%%%%%%%%%%%%%%%%%%%%%%%%%%%%%%%%%%%%%%%%%%%%%%%%%%%%%
% Section 3. Laplace transformation for Black-Scholes eq.
%%%%%%%%%%%%%%%%%%%%%%%%%%%%%%%%%%%%%%%%%%%%%%%%%%%%%%%%%%%%%%%%%%%%%%%%%%%%
\section{Laplace transformation method for the Black-Scholes equation}\label{sec:lap-bs}
In this section, we will apply the Laplace transformation method to the Black-Scholes
equation depending on one stock asset.
A basket option depending on several assets can be extended from the following
numerical scheme and analyzed in a similar way.
Taking Laplace transforms of \eqref{eq:bs-one}, we have
\be  \label{eq:bs-one-lap}
  z\hu
  -\frac12 \si^2x^2\frac{\rd^2 \hu}{\rd x^2}
  - rx\frac{\rd \hu}{\rd x} + r\hu=u_0, \quad (x,z)\in\BbbR_+\times\Gamma.
\ee
In what follows, we will examine the solvability of the above
equation and the resolvent of the Black-Scholes equation.

%\subsection{Solvability of Laplace transformed equation}
\subsection{The weak formulation of the Laplace transformed equation}
For a concrete mathematical analysis, we restrict our
attention to a European put option.
Since the boundary condition of a put option vanishes at infinity,
the partial differential equation can be reformulated as a weak problem
in a weighted Sobolev space.
Let $L^2(\BbbR_+)$ be the space of square integrable
complex-valued functions on $\BbbR_+$ which is endowed with the inner-product
$(v,w) =$ $\int_{\BbbR_+} v(x)\overline{w}(x)\, dx$ and the norm
$\|v\|_{L^2(\BbbR_+)} = \sqrt{(v,v)}.$ Then following \cite{comp-option},
the weighted Sobolev spaces are defined:
\begin{definition}
Let $\V$ be the weighted Sobolev space defined by
    \be
    \V=\{v\in L^2(\BbbR_+) :\,   x\frac{\rd v}{\rd x}\in L^2(\BbbR_+) \},  \nn
    \ee
%with the norm defined by
%$$
%\|v\|_\V=\Big( \|v\|_{L^2(\BbbR_+)}^2 + \|x\frac{\rd v}{\rd x}\|_{L^2(\BbbR_+)}\Big)^{\frac12},
%$$
equipped with the the semi-norm and the norm
$$
|v|_\V=\Big( \int_0^\infty \left| x\frac{\rd v}{\rd x}\right|^2 \,dx\Big)^{\frac12},\quad
\|v\|_\V=\Big( \int_0^\infty | v|^2 + \left|x \frac{\rd v}{\rd x}\right|^2 \,dx \Big)^{\frac12}.
$$
Similarly, let $\Z$ be the weighted Sobolev space defined by
    \be
    \Z=\{v\in L^\infty(\BbbR_+) : x\frac{\rd v}{\rd x}\in L^\infty(\BbbR_+) \},  \nn
    \ee
%with the norm defined by
%$$
%\|v\|_\V=\Big( \|v\|_{L^2(\BbbR_+)}^2 + \|x\frac{\rd v}{\rd x}\|_{L^2(\BbbR_+)}\Big)^{\frac12},
%$$
equipped with the the semi-norm and the norm
$$
|v|_\Z= \esssup_{x\in\BbbR_+} \left| x\frac{\rd v}{\rd x}\right|,\quad
\|v\|_\Z=\max\left\{\esssup_{x\in\BbbR_+} \left|v\right|,\,
\esssup_{x\in\BbbR_+} \left| x\frac{\rd v}{\rd x}\right|\right\}.
$$
\end{definition}
Since the boundary value vanishes at infinity, we have the following
Poincar\'{e}-type inequality,
which is an extension of the real-valued version, Lemma 2.7 given in \cite{comp-option}:
\begin{lemma}\label{lem:weighted-poincare}
The following bound holds:
\begin{eqnarray}
    \| v \|_{L^2(\BbbR_+)}\le 2 |v|_\V\qquad\forall v\in \V.
  \end{eqnarray}
\end{lemma}

\begin{proof}
Let $v\in \V$ be arbitrary. Then, by integration by parts,
the following relation holds:
    $$
    -\int_{\BbbR_+}v\overline{v}\,dx = \int_{\BbbR_+}xv\overline{v}_x\,dx +
    \int_{\BbbR_+}x\overline{v}v_x\,dx.
% (v_r + i v_i) (w_r - i w_i) +  (v_r - i v_i) (w_r + i w_i)
% = 2(v_r w_r + v_i w_i)
    $$
Thus we obtain
    $$
    \int_{\BbbR_+}|v|^2\,dx \le 2 \Big( \int_{\BbbR_+}|v|^2\,dx \Big)^{\frac12}\Big( \int_{\BbbR_+}|xv_x|^2\,dx \Big)^{\frac12}.
    $$
    This completes the proof.
\end{proof}

From now on, assume that the initial data $u_0\in \V'$, where $\V'$ is the dual
space of $\V$.
Denote by $\V^\prime$ the topological dual space of $\V$ with the norm defined
by
\[
\|u\|_{\V^\prime} = \sup_{v\in \V\setminus \{0\}} \frac{\<u,v\>}{\|v\|_\V},
\]
where $\<\cdot,\cdot\>$ is the duality pairing of $\V^\prime$ and $\V$.

Then, multiplying \eqref{eq:bs-one-lap} by a test function $v\in \V$ and
integrating on $\BbbR_+$,
one obtains the weak problem of \eqref{eq:bs-one-lap} as follows:
%Suppose $u_0\in \V^\prime.$
For each $z\in\G$, find $\hat u(z)\in \V$ such that
\begin{eqnarray}\label{eq:weakprob}
A_z(\hat u,v) = \<u_0,v\>\quad \forall v \in \V,
\end{eqnarray}
where the bilinear form $A_z(\cdot,\cdot): \V\times \V \rightarrow \mathbb C$
is defined by
\begin{eqnarray}
A_z(u,v) = z (u,v) + B(u,v) \quad \forall u,v\in \V.
\end{eqnarray}
where
\be\nonumber
B(u,v) = \frac12\int_{\BbbR_+}\sigma^2(x)x^2\frac{\rd u}{\rd
x}\frac{\rd \overline{v}}{\rd x}\,dx &+& \int_{\BbbR_+}\Big( -r(x) +
\sigma^2(x)+x\sigma(x)\frac{\rd \si}{\rd x}\Big)x\frac{\rd u}{\rd
x}\overline{v}\,dx \\
&&+\int_{\BbbR_+}r(x)u\overline{v}\,dx,
\nonumber
\ee
The bilinear form $A_z(\cdot,\cdot)$, of course, depends on $z \in \Gamma,$
and so does the solution $\hat u.$

\begin{assumption}\label{ass:bs-two-lap}
Assume that $\si\in \Z$ and $r\in L^\infty(\BbbR_+).$ Moreover, assume that
there exists a positive constant $\underline{\si}$
such that for all $x\in\BbbR_+$ such that
$$
0<\underline{\si}\le\si(x).
$$
\end{assumption}

Set
\[
\mu= \begin{cases}
  (\|r\|_{L^{\infty}(\BbbR_+)}-\si^2)^2/(\underline{\si})^2, & \text{
    if } \sigma(x) \text{ is a constant},\\
  (\|r\|_{L^{\infty}(\BbbR_+)}+2\|\si\|_\Z^2)^2/(\underline{\si})^2, & \text{
    otherwise.}
\end{cases}
\]
We now have the following two lemmas for the continuity and coercivity of
$A_z(\cdot,\cdot):\V\times \V\rightarrow \mathbb C$.
\begin{lemma}\label{lem:conti}
Under \assref{ass:bs-two-lap}, the bilinear form
$A_z(\cdot,\cdot):\V\times \V\rightarrow \mathbb C$ is continuous.
\end{lemma}
\begin{proof}
    Let $u,v\in \V.$ Then,
    \begin{align}
    \Big| \int_{\BbbR_+} \frac12\si^2(x)x^2\frac{\rd u}{\rd x}\frac{\rd
      \overline{v}}{\rd x}\,dx \Big| &\le \frac12 |\si|_\Z^2|u|_\V|v|_\V ,\nn \\
    \Big|\int_{\BbbR_+}\Big( -r(x) + \sigma^2(x)+x\sigma(x)\frac{\rd \si}{\rd
      x}\Big)x\frac{\rd u}{\rd x}\overline{v}\,dx\Big|
%&\le \left(\|r\|_{L^\infty(\BbbR_+)} +
%2\|\sigma\|_\Z^2\right)|u|_\V\|v\|_{L^2(\BbbR_+)}\nn \\
%    &\le 2 \left(\|r\|_{L^\infty(\BbbR_+)} + 2\|\sigma\|_\Z^2\right)|u|_\V|v|_\V, \nn \\
&\le \underline{\si} \sqrt{\mu} \, |u|_\V\|v\|_{L^2(\BbbR_+)}\nn \\
    &\le 2 \underline{\si} \sqrt{\mu}\,    |u|_\V|v|_\V, \nn \\
    \Big|\int_{\BbbR_+}\Big(z+r(x)\Big)u\overline{v}\,dx\Big| & \le (|z|+\|r\|_{L^\infty(\BbbR_+)})|u|_\V|v|_\V \nn,
    \end{align}
where \lemref{lem:weighted-poincare} is applied in the bound of the second inequality.
Therefore the bilinear form $A_z(\cdot,\cdot):\V\times \V\rightarrow \mathbb
C$ is continuous.
\end{proof}

\begin{lemma}\label{lem:coer}
Under \assref{ass:bs-two-lap}, there is a non-negative constant $C_1$, which
is independent of $u$ and $z$, such that for all $u\in \V$
    $$
    \Re\{A_z(u,u)\} \ge \frac{\underline{\sigma}^2}{4}|u|^2_\V-(|z|+C_1)\|u\|^2_{L^2(\BbbR_+)}.
    $$
\end{lemma}
\begin{proof}

Under \assref{ass:bs-two-lap}, the following result is known in
\cite{comp-option},
\be 
\Re\{B(u,u)\} \ge
\frac{\underline{\sigma}^2}{4}|u|^2_\V-\mu\|u\|^2_{L^2(\BbbR_+)}.
\label{eq:coercive} 
\ee 
%where
%$\mu=(\|r\|_{L^{\infty}(\BbbR_+)}+2\|\si\|_\Z^2)^2/(\underline{\si})^2$.
    Let $u\in \V$ be arbitrary. Then,
    \begin{align}
    \int_{\BbbR_+} \frac12\si^2(x)x^2\frac{\rd u}{\rd x}\frac{\rd \overline{u}}{\rd x}\,dx &\ge \frac{\underline{\si}^2}{2}|u|_\V^2 , \nn\\
    \Big|\Re\{\int_{\BbbR_+}\Big( -r(x) + \si^2(x)+x\si(x)\frac{\rd \si}{\rd
      x}\Big)x\frac{\rd u}{\rd x}\overline{u}\,dx\}\Big|
%& \le     \left(\|r\|_{L^\infty(\BbbR_+)} + 2\|\sigma\|_\Z^2\right)
%|u|_\V\|u\|_{L^2(\BbbR_+)} \nn\\
& \le \underline{\si}  \sqrt{\mu}\,  |u|_\V\|u\|_{L^2(\BbbR_+)} \nn\\
    &\le \frac{\underline{\si}^2}{4}|u|^2_\V + \mu\|u\|^2_{L^2(\BbbR_+)},&  \nn\\
    \Big|\Re\{\int_{\BbbR_+}\Big(z+r(x)\Big)u\overline{u}\,dx\}\Big| & \le
    (|z|+ \|r\|_{L^\infty(\BbbR_+)}  )\|u\|^2_{L^2(\BbbR_+)},\nn
%\le\frac{\mu}{2}\|u\|^2_{L^2(\BbbR_+)},\nn
    \end{align}
where Young's inequality is used in the bound of the second inequality and
$\mu$ depends on $\underline{\si}$.
A combination of these inequalities completes the lemma.
\end{proof}
The compactness of embedding $L^2(\mathbb R_+)\hookrightarrow \V$,
\lemref{lem:conti} and \lemref{lem:coer} imply that there is a unique solution
in the case of European put options. We summarize the above results in the
following theorem.
%with a given function data $u_0(x)\in L^2(\BbbR_+)$.
\begin{theorem} Suppose $u_0\in \V'$. Then,
under Assumption \ref{ass:bs-two-lap}
Problem \eqref{eq:weakprob} has a unique solution $\hu\in \V.$
\end{theorem}

%The above argument does not apply to European call options, but the
%solvability can also be proved in the truncated domain as is described
%in the following section.

\subsection{Resolvent of the Black-Scholes equation}
In \secref{sec:laplace}, the resolvent of a spatial
operator is assumed to be bounded in a given sector. This assumption for the Black-Scholes 
equation will be verified in this subsection.

Denote by $R(z,-B)=(zI+B)^{-1}$  the resolvent of $-B$, so that
for each $f\in \V^{\prime}$, $v = R(z,-B)f$ is the solution of
\be
B(v,\phi)+z(v,\phi)=\<f,\phi\>,\qquad \forall \phi\in \V. \label{eq:resolvent-eq}
\ee
Then we have the following lemma, which is an extension of Lemma 2.1
in \cite{resolvent-one-dim}.
\begin{lemma}\label{le:kappa}
Under \assref{ass:bs-two-lap}, for any
$\theta\in(\frac12\pi,\pi)$ there are $C = C(\theta)\ge 0$ and
$\kappa=\kappa(\theta, r, \sigma)>0$, independent of
$z$ and $f$, such
that
$$
\|R(z,-B)f\|_{L^2({\BbbR_+})}\le\frac{C}{|z-\kappa|}\|f\|_{L^2({\BbbR_+})},\qquad \text{for }z\in\Sigma_{\kappa,\theta},\,\,f\in L^2({\BbbR_+})
$$
where
$\Sigma_{\kappa,\theta}=\{z\in\mathbb{C}:|\arg(z-\kappa)|\le\theta\}$.
Explicitly, the coefficients are given by
$C = (1+\frac12\delta)(1+\delta^2)$ and 
$\kappa= \left(1 + \frac{\delta^2}{2}\right) \mu,$
where $\delta=\tan\frac{\theta}{2}.$
\end{lemma}
\begin{proof}
For $z\in\Sigma_{\kappa,\theta}$, we write
$$
z-\kappa=(\xi+\ii\eta)^2=\xi^2-\eta^2+2\ii\xi\eta\qquad\text{with}\,\,\xi+\ii\eta\in\Sigma_{0,\theta/2},\,\,\xi,\eta\in\BbbR,
$$
for any $\kappa >0.$
Setting 
$\delta=\tan\frac{\theta}{2}$, we see that $\delta > 1$ and
$|\eta|\le\delta\xi.$
and thus the following inequality holds:
$$
\xi^2\le|z-\kappa|=\xi^2+\eta^2\le(1+\delta^2)\xi^2.
$$
Set
$$
F = B(v,v) + z\|v\|_{L^{2}(\BbbR_+)}^2.
$$
Taking the real part of $F$, we obtain
$$
\Re B(v,v)+(\kappa+\xi^2-\eta^2)\|v\|_{L^{2}(\BbbR_+)}^2=\Re F.
$$
By the inequality \eqref{eq:coercive} we have 
\be
\frac{\underline{\sigma}^2}{4}|v|_\V^2+(\kappa+\xi^2-\eta^2-\mu)\|v\|_{L^{2}(\BbbR_+)}^2\le|F|.\label{eq:real-F}
\ee 
By taking the imaginary part of $F$, we have
$$
\Im B(v,v)+2\xi\eta\|v\|_{L^{2}(\BbbR_+)}^2=\Im F,
$$
and since $\Im B(v,v)=\Im \int_{\BbbR_+}\Big( -r(x) +
\sigma^2(x)+x\sigma(x)\frac{\rd \si}{\rd x}\Big)x\frac{\rd v}{\rd
x}\overline{v}\,dx$,
$$
2\xi|\eta|\,\|v\|_{L^{2}(\BbbR_+)}^2\le|F|+
%(\|r\|_{L^{\infty}(\BbbR_+)}+2\|\si\|_\Z^2)|v|_\V\|v\|_{L^{2}(\BbbR_+)}.
 \underline{\si} \sqrt{\mu} \,|v|_\V\|v\|_{L^{2}(\BbbR_+)}.
$$
Multiplying by $\frac12\delta=\frac12\tan(\frac12\theta)$ the last estimate, 
we have
$$
\eta^2\|v\|_{L^{2}(\BbbR_+)}^2\le\delta|\eta|\,\|v\|_{L^{2}(\BbbR_+)}^2
%\le\frac12\delta|F|+\frac12\delta(\|r\|_{L^{\infty}(\BbbR_+)}+2\|\si\|_\Z^2)|v|_\V\|v\|_{L^{2}(\BbbR_+)}.
\le\frac12\delta|F|+\frac12\delta\,  \underline{\si} \sqrt{\mu} \,|v|_\V\|v\|_{L^{2}(\BbbR_+)}.
$$
Adding this to \eqref{eq:real-F}, we obtain
$$
\frac{\underline{\sigma}^2}{4}|v|_\V^2+(\kappa+\xi^2-\mu)\|v\|_{L^{2}(\BbbR_+)}^2
\le(1+\frac12\delta)|F|+\frac{\underline{\sigma}^2}{8}|v|_\V^2
%+\frac{\delta^2(\|r\|_{L^{\infty}(\BbbR_+)}+2\|\si\|_\Z^2)^2}{2\underline{\sigma}^2}\|v\|_{L^{2}(\BbbR_+)}^2.
%+\frac{\delta^2 \mu }{2\underline{\sigma}^2}\|v\|_{L^{2}(\BbbR_+)}^2.
+\frac{\delta^2 \mu }{2}\|v\|_{L^{2}(\BbbR_+)}^2.
$$
%Recalling
%$\mu=(\|r\|_{L^{\infty}(\BbbR_+)}+2\|\si\|_\Z^2)^2/(\underline{\si})^2$.
With the choice of
$$
%\kappa=\mu+\frac{\delta^2(\|r\|_{L^{\infty}(\BbbR_+)}+2\|\si\|_\Z^2)^2}{2\underline{\sigma}^2}=\frac{(2+\delta^2)(\|r\|_{L^{\infty}(\BbbR_+)}+2\|\si\|_\Z^2)^2}{2\underline{\sigma}^2},
\kappa=\mu+\frac{\delta^2 \mu}{2}= \left(1 + \frac{\delta^2}{2}\right) \mu,
$$
we have the following inequality
$$
\frac{\underline{\sigma}^2}{8}|v|_\V^2+\xi^2\|v\|_{L^{2}(\BbbR_+)}^2\le(1+\frac12\delta)|F|.
$$
If $f\in L^2(\BbbR_+)$, we take $\phi=v$ in \eqref{eq:resolvent-eq}, then we have
$$
\frac{\underline{\sigma}^2}{8}|v|_\V^2+\xi^2\|v\|_{L^{2}(\BbbR_+)}^2\le(1+\frac12\delta)\Big|\int_{\BbbR_+}fv\,dx\Big|\le(1+\frac12\delta)\|f\|_{L^2(\BbbR_+)}\|v\|_{L^2(\BbbR_+)},
$$
and therefore
$$
\|R(z,-B)f\|_{L^{2}(\BbbR_+)}\le\frac{1+\frac12\delta}{\xi^2}\|f\|_{L^{2}(\BbbR_+)}\le\frac{(1+\frac12\delta)(1+\delta^2)}{|z-\kappa|}\|f\|_{L^{2}(\BbbR_+)}.
$$
This completes the proof.
\end{proof}
From this lemma one can determine the location of a integration contour. In particular,
if one sets the asymptotic slope of a hyperbola as $s$, the contour has 
to cut the real line at a point which is larger than
$$
\kappa = \left( 1+ \frac{\tan^2(\frac12\arctan(s))}{2}\right)\,\mu .
$$
In the special cast that $r(x)=r$ and $\sigma(x)=\sigma$ are constants, 
$\kappa$ can be given by
\be
\kappa =  \left( 1+ \frac{\tan^2(\frac12\arctan(s))}{2}\right)\,
\frac{|r-\sigma^2|^2}{\sigma^2}.\label{eq:kappa-con}
\ee

\subsection{The transparent boundary condition}
As one can see in \eqref{eq:bs-one} or \eqref{eq:bs-basket}, the space domain
of the underlying asset of an option is an unbounded set.
To apply a numerical scheme, one usually truncates the infinite domain into a finite one,
%at large values of the asset price and suitable boundary conditions are applied. The following
and then imposes a suitable boundary condition on the boundary.
Let $L$ be a sufficiently large asset price.
One then has the following version of the Black-Scholes equation  truncated at $x=L$.
\begin{eqnarray}%\label{eq:bs-one-trun}
 \frac{\rd u}{\rd t}-\frac12{\si^2x^2}\frac{\rd^2u}{\rd x^2}-rx\frac{\rd u}{\rd x} + ru &=& 0, \quad (x,t)\in (0,L)\times(0,T], \label{eq:bs-one-truna} \\
    u(x,t)&=&g(x,t), \quad (x,t)\in \rd (0,L) \times(0,T], \label{eq:bs-one-trunb} \\
    u(x,0)&=&u_0, \quad x\in[0,L]. \label{eq:bs-one-trunc}
\end{eqnarray}
In many cases, the boundary condition on the artificial boundary $x=L$ is imposed by
extending a given payoff function. For example, European put
options assume $u(L,t)=0$ and European call options assume $\frac{\rd u}{\rd x}(L,t)=1$.
In \cite{trun-bd}, the errors caused by Dirichlet boundary
conditions on the artificial boundary are estimated and thus one can determine a
suitable truncation asset price for the artificial boundary
to meet a given error tolerance.

Instead of such artificial boundary conditions,
a transparent boundary condition is introduced in \cite{comp-option}
with which one can evaluate the solution in the truncated domain without any
truncation error.
However, the boundary condition in \cite{comp-option} is an
integro-differential one, which needs some suitable numerical schemes to
approximate it that will produce other possibly significant errors.
We will analyze the transparent boundary condition in more detail and then depart from
such an integro-differential type, by implementing the boundary condition in
the Laplace transformed setting instead of the usual space-time setting. 
Our transparent boundary condition is motivated by the following proposition.
\begin{proposition}\label{prop:trans-bc}
Assume that the coefficients $\si$ and $r$ in \eqref{eq:bs-one-lap} are
constants
and that $L>0$ is sufficiently large so that $\supp(u_0) \subset [0, L)$. Then
among the solutions $\hu(x,z)$ satisfying \eqref{eq:bs-one-lap}
there is a component, $\hu_+$, satisfying the following:
\begin{align}\label{eq:exterior-eq}
  \frac{\p\hu_+}{\p x}(x,z)=\frac{1}{x\si^2}\left\{ -(r-\frac12\si^2) -
\sqrt[+]{\big(r-\frac12\si^2\big)^2+2\si^2(r+z) }
  \right\}\hu_+(x,z)
\end{align}
$ \quad\forall x \in (L,\infty),$ where $\Re\{\sqrt[+]{z}\}>0$ for nonzero $z\in\mathbb C$.
\end{proposition}
\begin{proof}
Take the change of variables, $y=\log x$, to \eqref{eq:bs-one-lap}.
Denoting by $\hv$ its solution, owing to $\supp(u_0) \in [0,L)$, 
one observes that $\hv$ satisfies the right exterior problem
\be
  z\hv
  -\frac12 \si^2\frac{\rd^2 \hv}{\rd y^2}
  - (r-\frac12\sigma^2)\frac{\rd \hv}{\rd y} + r\hv=0, \quad (y,z)\in  (L,\infty)\times\Gamma.
%
% \frac{\rd v}{\rd t}-\frac12\si^2\frac{\rd^2v}{\rd y^2}-(r-\frac12\si^2)\frac{\rd v}{\rd y} + rv = 0, \quad (y,t)\in (\log L,\infty)\times(0,T], \nn %\\
  %  v(\log L,t)=0\,\,\textrm{and}\,\, v(\infty,t)=0,\quad t\in(0,T],\nn \\
  %  v(y,0)=0, \quad y\in[\log L,\infty).\nn
\ee
Among the two linearly independent solutions, we take the component which
vanishes at infinity, which is given as follows:
%Using the Laplace transformation in $t$-variable, we can solve explicitly $\hv=\mathcal{L}(v)$ like following
$$
\hv_+(y,z) = \exp\Big( \left\{\frac{-(r-\si^2/2)}{\si^2} - \frac{1}{\si^2}\sqrt[+]{\big(r-\frac12\si^2\big)^2+2\si^2(r+z) }\right\}y \Big).
$$
Restoring the change of variable, $x=e^y$, and denoting by $\hu_+(x)=\hv_+(y)$, one gets
$$
\hu_+(x,z) = x^{\left\{\frac{-(r-\si^2/2)}{\si^2} - \frac{1}{\si^2}\sqrt[+]{\big(r-\frac12\si^2\big)^2+2\si^2(r+z) }\right\}}.
$$
Thus, by differentiating with respect to $x$, one arrives at
$$
\frac{\p\hu_+}{\p x}(x,z) =\frac1{x\sigma^2} \left\{ -(r-\si^2/2) - \sqrt[+]{\big(r-\frac12\si^2\big)^2+2\si^2(r+z) }\right\}\hu_+(x,z).
$$
Thus, $\hu_+$ satisfies the equation \eqref{eq:exterior-eq},
which completes the proof.
\end{proof}
Due to \propref{prop:trans-bc}, by choosing $L>0$ sufficiently large so that
$\supp(u_0)\in [0,L),$ we propose the following transparent boundary
condition
at $x=L:$
\begin{align}\label{eq:trans-bc}
  \frac{\p\hu}{\p x}(L,z)=\frac{1}{L\sigma^2}\left\{-(r-\frac12\si^2)-
\sqrt[+]{\big(r-\frac12\si^2\big)^2+2\si^2(r+z) }
  \right\}\hu(L,z) \quad\forall z \in \G.
\end{align}

\begin{remark}
By the Laplace inversion of \eqref{eq:trans-bc},  the transparent
boundary condition  in the space-time domain is given by
\begin{align}\label{eq:trans-bc-time}
\frac{\p u}{\p x}(L,t)=\frac{1}{L\sigma^2}\left\{ -(r-\frac{\si^2}{2})u(L,t)
- \frac{\sqrt{2}\si}{\sqrt{\pi}}e^{-\eta t} \frac{\rd}{\rd t}\int^t_0\frac{u(L,\tau)e^{\eta\tau}}{\sqrt{t-\tau}}d\tau \right\},
\end{align}
where $\eta=\frac{(r-\si^2/2)^2}{2\si^2}+r$. In the derivation of 
\eqref{eq:trans-bc-time}, the following equalities are used:
\bes
  \mathcal{L} \left\{ \frac{\rd}{\rd t}\int^t_0 \frac{1}{ \sqrt{t-\tau} } u(\tau)e^{\eta\tau} d\tau \right\} &=&z\mathcal{L} \left\{ \int^t_0 \frac{1}{ \sqrt{t-\tau} } u(\tau)e^{\eta\tau} d\tau \right\} \nn \\
&=& z\mathcal{L} \left\{ \frac{1}{\sqrt{t}} \right\}\mathcal{L} \left\{ u(t)e^{\eta t} \right\} \nn\\
&=& \sqrt{\pi}\sqrt{z}\hu(z-\eta). \nn
\ees
In solving the partial integro-differential equation
\eqref{eq:exterior-eq} with \eqref{eq:trans-bc-time} using a Crank-Nicolson
type of time-marching algorithm, one usually needs an
expensive algorithm in computing time and memory. We will compare our
Laplace transformation method with the Crank-Nicolson method in \secref{sec:num},
and conclude superiority in using our method.
%must be approximated by applying some kinds of numerical
%scheme to the integration in the above condition,
%and thus it produces numerical error from
%the boundary truncation. In our proposed method, however, this condition can be
%imposed naturally during solving the equation in \eqref{eq:bs-one-lap}
%like an ordinary robin boundary condition,
%and therefore it can be easily implemented and evaluated without any error.
\end{remark}

\section{Numerical results}\label{sec:num}
We applied the Laplace transformation method for time discretization while the
standard piecewise linear ($P_1$) finite element method for the space
discretization is used.
%in \exrefs{ex:p-1}{ex:p-2}, and the 5-point finite difference 
%method in \exref{ex:p-3}.
Using an analytic solution for the first two examples, we can compare the
convergence rate of the proposed scheme. 
In \exref{ex:p-2} we examine the effects of
the Dirichlet boundary condition and the transparent boundary condition 
\eqref{eq:trans-bc} in the calculation of option prices.

In calculating the numerical values of the analytical solution, the error
function $erf(x)$ is evaluated by using the algorithm on page 213 of Numerical
Recipes in Fortran \cite{recipe-fortran} which has 16-digit precision.
The reduction rate and speedup are defined by
$$
\textrm{reduction rate}=\log_2\frac{\|u_{\Delta x}-u_{\textrm{exact}}\|_{L^2(0,L)}}{\|u_{\frac{\Delta x}{2}}-u_{\textrm{exact}}\|_{L^2(0,L)}},
$$
where $u_{\Delta x}$ denotes the numerical solution with the spatial mesh size
$\Delta x$,
and
$$
\textrm{speed up}=\frac{\textrm{time consumption}}{\textrm{time consumption using 1-CPU}} .
$$
\begin{example}[European put option with constant coefficients]\label{ex:p-1}
We consider an European put option with coefficients
$r=0.05$, $\sigma=0.3$, $T=1.0$ and $K=50$ and we truncate the domain at $L=200$.
\end{example}
For the numerical solutions, the boundary condition at $x=0$ in \eqref{eq:bs-one-trunb} is given by
\[
u(x,t) = K e^{-rt}, \quad (x,t)\in \{0\}\times (0,T],
\]
while that at $x=L$
\begin{equation}\label{eq:p1-diri-bc}
u(x,t) = 0, \quad (x,t)\in \{L\}\times (0,T].
\end{equation}
Although an analytic solution to this example is given by Black and Scholes,
it is our aim to compare convergence rates for the proposed scheme and
the standard time-marching algorithms such as %the backward-Euler and
Crank-Nicolson scheme. %s.
%\tabrefs{tab:p-1-im}{tab:p-1-cn}
\tabref{tab:p-1-cn}
%show convergence rates for the backward-Euler and Crank-Nicolson schemes, respectively. 
shows convergence rate for the Crank-Nicolson scheme. 
%As can be expected, they give first-order and second-order convergence rates.
As can be expected, it gives first-order convergence rate.
\tabref{tab:p-1-lap} shows that the choice of
15 $z$-points in the contour with the proposed method is enough to obtain the same level of tolerance
attained using 640 time steps with the Crank-Nicolson method.
%and 40960 time steps with the backward-Euler method.
Observe that
for each $z$-point the cost of solving the complex-valued elliptic
problem using the proposed method is almost comparable to that of advancing
one step forward by solving the real-valued elliptic problem with the time-marching
%in applying the time-marching 
algorithms.

In the proposed scheme, we need the value of $\kappa$ as in \lemref{le:kappa} to determine the location
of a integration contour. Since the coefficients are constants, if we choose
the asymptotic slope of the contour as $0.4$, 
we have $\kappa = 0.01811$ by \eqref{eq:kappa-con}, and therefore the contour
has to cut the real line at a point greater than 0.01811.
Under this constraint, we
choose the optimal parameters which are suggested in \cite{contour-para},
and these parameters are attached in \tabref{tab:p-1-con} in the case that
the evaluation time is 1.0 for different iteration numbers.
In particular, \tabref{tab:p-1-con} says that 12 iterations are enough to
balance with the space discretization of 2560 spatial meshes.
%delete according to the recommendation of referee
%
%In \figref{fig:p-1} we plot the numerical result to compare 
%with the exact solution for Example 4.1.

%\begin{center}
%\begin{table}[htbp]
%\centering
%\begin{tabular}{|c|x{2.4cm}|c|c|c|}\hline
%        Time steps& Number of space meshes & Mesh size& Error in $L^2$ & Reduction rate \\\hline
%      10 & 10 &20 & 2.769 & \\\hline
%      40 & 20  &10 & 0.6992 & 1.985\\\hline
%      160 & 40  &5 & 0.1740 & 2.007 \\\hline
%      640 & 80  &2.5 & 0.4347E-01 & 2.001 \\\hline
%      2560 & 160  &1.25 & 0.1087E-01 & 2.000 \\\hline
%      10240 & 320  &5/8 & 0.2716E-02 & 2.000 \\\hline
%      40960 & 640  &5/16 & 0.6792E-03 & 2.000 \\\hline
%\end{tabular}
%\caption{\exref{ex:p-1} with the backward-Euler method}
%\label{tab:p-1-im}
%\end{table}
%\end{center}

\begin{table}[htbp]
\centering
\begin{tabular}{|c|x{2.4cm}|c| c|c|}\hline
        Time steps&Number of space meshes & Mesh size & Error in $L^2$ & Reduction rate \\\hline
           10 & 10 &20  & 2.928 & \\\hline
      20 & 20 &10  & 0.7536 & 1.958 \\\hline
      40 & 40  &5 & 0.1878 & 2.004 \\\hline
      80 & 80 &2.5 & 0.4695E-01 & 2.000 \\\hline
      160 & 160  &1.25 & 0.1174E-01 & 2.000 \\\hline
      320 & 320  &5/8 & 0.2934E-02 & 2.000 \\\hline
      640 & 640 &5/16 & 0.7337E-03 & 2.000 \\\hline
\end{tabular}
\vspace{.3cm}
\caption{\exref{ex:p-1} with the Crank-Nicolson method}
\label{tab:p-1-cn}
\end{table}

\begin{table}[htbp]
\centering
\begin{tabular}{|c|x{2.4cm}|c|c|c|}\hline
        Number of $z$ &Number of space meshes & Mesh size & Error in $L^2$ & Reduction rate \\\hline
          15 & 10 &20  & 2.924 &\\\hline
      15 & 20 &10 & 0.7524 & 1.959\\\hline
      15 & 40  &5 & 0.1876 & 2.004\\\hline
      15 & 80  &2.5& 0.4688E-01 & 2.000\\\hline
      15 & 160  &1.25 & 0.1172E-01 & 2.000\\\hline
      15 & 320  &5/8 & 0.2930E-02 & 2.000\\\hline
      15 & 640 &5/16 & 0.7327E-03 & 2.000\\\hline
\end{tabular}
\vspace{.3cm}
\caption{\exref{ex:p-1} with the Laplace transformation method}
\label{tab:p-1-lap}
\end{table}

\begin{table}[htbp]
\centering
\begin{tabular}{|x{1.4cm}|x{1.6cm}|c|x{1.6cm}|c|c|c|c|}\hline
Number of $z$ &Number of space meshes &$L^2$-Error & Reduction rate & $\gamma$ & $\nu$ & $s$ & $\tau$ \\\hline
3 & 2560 & 0.6397E-00 &       & 13.48 & 12.42 & 0.4213 & 0.16500\\\hline
6 & 2560 & 0.1705E-01 & 5.229 & 26.95 & 24.84 & 0.4213 & 0.09385\\\hline
9 & 2560 & 0.3434E-03 & 5.634 & 40.43 & 37.26 & 0.4213 & 0.06809\\\hline
12& 2560 & 0.5642E-04 & 2.605 & 53.90 & 49.68 & 0.4213 & 0.05430\\\hline
15& 2560 & 0.4731E-04 & 0.003 & 67.38 & 62.09 & 0.4213 & 0.04556\\\hline
18& 2560 & 0.4721E-04 & 0.001 & 80.86 & 74.51 & 0.4213 & 0.03947\\\hline
21& 2560 & 0.4717E-04 & 0.000 & 94.33 & 86.93 & 0.4213 & 0.03494\\\hline
%15 & 67.38 & 62.09 & 0.4213 &  0.04556 \\\hline
%30 & 134.8 & 124.2 & 0.4213 & 0.02278\\\hline
\end{tabular}
\vspace{.3cm}
\caption{Contour Parameters for \exref{ex:p-1}}
\label{tab:p-1-con}
\end{table}

%delete according to the recommendation of referee
%
%\begin{figure}
%\begin{center}
%  \resizebox{100mm}{!}{\includegraphics{./put-time-1.0.eps}}
%  \caption{Comparison with exact solution of \exref{ex:p-1} }
%  \label{fig:p-1}
%\end{center}
%\end{figure}

\begin{example}[European put option with transparent boundary condition]\label{ex:p-2}
We consider a European put option with coefficients
$r=0.05$, $\sigma=0.3$, $T=1.0$ and $K=50$ and we truncate the domain at $L=50$.
\end{example}
In this example, we truncate the domain at the strike price, and then we replace
the Dirichlet boundary condition \eqref{eq:p1-diri-bc} with 
the transparent boundary condition given in \eqref{eq:trans-bc}. An identical
contour as in the previous example has been adopted.
\tabref{tab:p-2-diri} shows that the Dirichlet boundary condition with the domain truncation 
makes a significant error,
which cannot be overcome by mesh refinement.
\tabref{tab:p-2-trans}, however, gives second order convergence which is shown in
\tabref{tab:p-1-lap} although its domain is much smaller
 than that for \exref{ex:p-1}. Indeed, comparing the same mesh sizes in
\tabref{tab:p-2-trans} 
and \tabref{tab:p-1-lap}, one can observe the numerical values are almost identical.
In \figref{fig:p-2}
we can see the difference between the transparent boundary condition and the
Dirichlet boundary. 
\begin{table}[htbp]
\centering
\begin{tabular}{|c|x{2.4cm}|c|c|c|}\hline
        Number of $z$ & Number of space meshes & Mesh size & Error in $L^2$ & Reduction rate \\\hline
                15&  10&  5&  10.35& \\\hline
                15&  20&  2.5& 10.40& -0.007 \\\hline
                15&  40&  1.25& 10.41& -0.002 \\\hline
                15&  80&  5/8 & 10.42&  0.000 \\\hline
                15& 160&  5/16 & 10.42&  0.000 \\\hline
                15& 320&  5/32 & 10.42&  0.000 \\\hline
                15& 640&  5/64 & 10.42&  0.000 \\\hline
\end{tabular}
\vspace{.3cm}
\caption{\exref{ex:p-2} with the Dirichlet boundary condition at $L=50$}
\label{tab:p-2-diri}
\end{table}

\begin{table}[htbp]
\centering
\begin{tabular}{|c|x{2.4cm}|c|c|c|}\hline
        Number of $z$ & Number of space meshes & Mesh size & Error in $L^2$ & Reduction rate \\\hline
                15&  10& 5&  0.1870& \\\hline
                15&  20& 2.5&  0.4656E-01&  2.006\\\hline
                15&  40& 1.25&  0.1163E-01&  2.001\\\hline
                15&  80& 5/8 & 0.2907E-02&  2.000\\\hline
                15& 160& 5/16 & 0.7267E-03&  2.000\\\hline
                15& 320& 5/32 & 0.1817E-03&  1.999\\\hline
                15& 640& 5/64 & 0.4551E-04&  1.998\\\hline
\end{tabular}
\vspace{.3cm}
\caption{\exref{ex:p-2} with the transparent boundary condition
  \eqref{eq:trans-bc}  at $L=50$}
\label{tab:p-2-trans}
\end{table}

\begin{figure}
\begin{center}
  \resizebox{100mm}{!}{\includegraphics{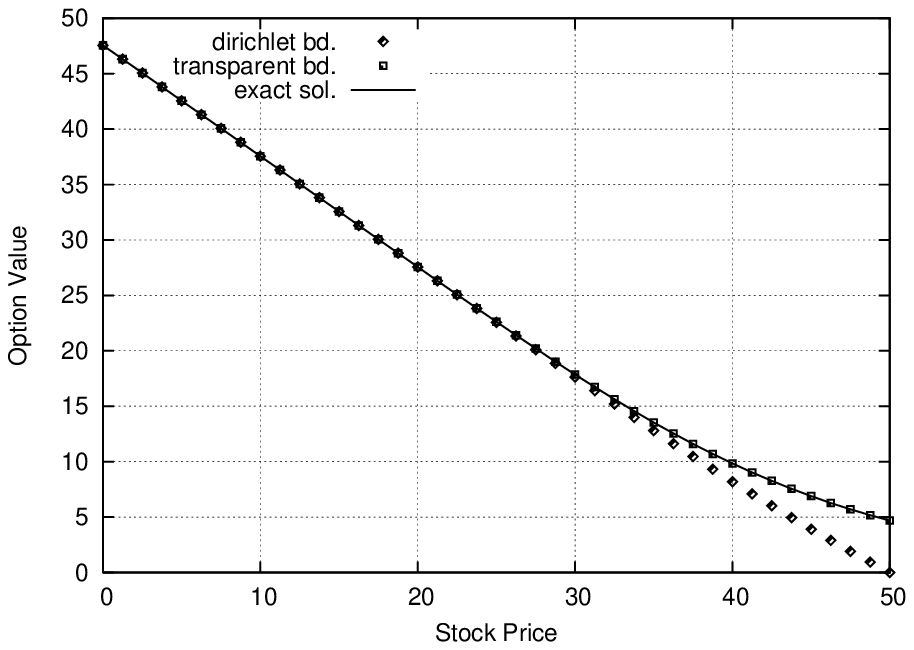}}
  \caption{Comparison between the Dirichlet boundary condition and the
    transparent boundary condition
\eqref{eq:trans-bc}
 in \exref{ex:p-2} }
  \label{fig:p-2}
\end{center}
\end{figure}

%\begin{figure}
%  \begin{center}
%	\subfigure[Option value]
%      {\resizebox{!}{50mm}{\includegraphics{./data/put-bd-time-1.0.eps}}}
%    \mbox{\subfigure[Delta]
%      {\resizebox{!}{50mm}{\includegraphics{./data/put-delta-time-1.0.eps}}}\quad
%    \subfigure[Gamma]
%      {\resizebox{!}{50mm}{\includegraphics{./data/put-gamma-time-1.0.eps}}}
%    }
%  \end{center}
%  \caption{Comparison of Greeks in Problem 2 }
%\end{figure}

\begin{example}[Basket option with two underlying assets]\label{ex:p-3}
We consider a European put basket option with two underlying assets having coefficients
$r=0.05$, $a_{11}=0.09$, $a_{22}=0.09$, $a_{12}=a_{21}=-0.018$, time to maturity=1.0,
artificial boundary $L_1=300, L_2=300$ and payoff function $(100-\max(x_1,x_2))_+$ is given.
\end{example}
For numerical computation, the boundary conditions are given by 
\begin{eqnarray}
\frac{\rd u}{\rd \nu}(\mathbf{x},t)=0, 
\quad\textrm{for } (\mathbf{x},t)\in \big(\{0\}\times(0,L_2)\cup(0,L_1)\times\{0\}\big)\times(0,T], \nn\\
u(\mathbf{x},t)=0, 
\quad\textrm{for } (\mathbf{x},t)\in \big(\{L_1\}\times(0,L_2)\cup(0,L_1)\times\{L_2\}\big)\times(0,T].\label{eq:basket-diri}
\end{eqnarray}
To evaluate the convergence rates for the proposed scheme, we solve the same
problem using the Crank-Nicolson scheme on a $512\times512$ space grid for the extended 
artificial domain $L_1=L_2=600$ with $\Delta t=0.02$.
We set this as the reference solution and calculate 
the relative $L^2$ error for the proposed scheme.
The integration contour is built using the parameters
$\gamma=35.94, \nu=33.12, s=0.4213, \tau=0.07472$.
%gamma=35.9358281308034      tau=  7.471771872785041E-002 nu=
%   33.1173051412011      slope=  0.421257354269125
Numerical results in \tabref{tab:bs-basket-con} show an almost second-order convergence rate.

\begin{table}[htbp]
\centering
    \begin{tabular}{|c|x{2.4cm}|c|x{2.5cm}|c|}\hline
        Number of $z$ & Number of space meshes & Mesh size & Relative error in $L^2$ & Reduction rate \\\hline
                15&  $16\times16$& 75/4 &   0.3662E-01& \\\hline
                15&  $32\times32$& 75/8 &    0.1047E-01&  1.806\\\hline
                15&  $64\times64$& 75/16 &    0.2969E-02&  1.819\\\hline
                15&  $128\times128$& 75/32 &  0.8444E-03&  1.814\\\hline
      \end{tabular}
\vspace{.3cm}
\caption{Convergence rate in \exref{ex:p-3} on the domain $[0,300]\times[0,300]$ }
\label{tab:bs-basket-con}
\end{table}

To shorten the artificial boundary, we apply the  
transparent boundary condition by assuming that the 
tangential derivative is negligible on the boundary.
Then the boundary condition \eqref{eq:basket-diri} is replaced with
$$
  \frac{\p\hu}{\p x_1}(L_1,x_2,z)=\frac{1}{L_1a_{11}}\left\{-(r-\frac12a_{11})-
\sqrt[+]{\big(r-\frac12a_{11}\big)^2+2a_{11}(r+z) }
  \right\}\hu(L_1,x_2,z) \nn
$$
on $(x_1,x_2,z)\in\{L_1\}\times(0,L_2)\times\G$, and
$$
  \frac{\p\hu}{\p x_2}(x_1,L_2,z)=\frac{1}{L_2a_{22}}\left\{-(r-\frac12a_{22})-
\sqrt[+]{\big(r-\frac12a_{22}\big)^2+2a_{22}(r+z) }
  \right\}\hu(x,L_2,z) \nn
$$
on $(x_1,x_2,z)\in(0,L_1)\times\{L_2\}\times\G$. 
In \tabref{tab:bs-basket-bd} we compare the results produced by
the different boundary conditions on the lines $L_1=150$ and $L_2=150$. 
As can be seen in \tabref{tab:bs-basket-bd}, the transparent boundary 
condition is more accurate than the Dirichlet boundary condition.
Furthermore, \tabref{tab:bs-basket-con} and \tabref{tab:bs-basket-bd} show
that if we apply the transparent boundary condition, it gives competitive  
error level in comparison to the Dirichlet boundary condition
even though its computational domain is a quarter size of that with
the Dirichlet boundary conditions applied.

\begin{table}[htbp]
\centering
    \begin{tabular}{|x{1cm}|x{2.4cm}|c|x{3.cm}|p{3.cm}|}\hline
Number of $z$ & Number of space meshes & Mesh size & Relative error in $L^2$(Dirichlet)& Relative error in $L^2$(Transparent) \\\hline
                15&  $16\times 16$ & 75/8 &   0.1998E-01 & 0.1076E-01 \\\hline
                15&  $32\times 32$ & 75/16 &    0.1176E-01 & 0.3485E-02 \\\hline
                15&  $64\times 64$ &  75/32 &   0.9283E-02 & 0.1724E-02 \\\hline
      \end{tabular}
\vspace{.3cm}
\caption{Effect of boundary conditions in \exref{ex:p-3} on the domain $[0,150]\times[0,150]$}
\label{tab:bs-basket-bd}
\end{table}

Since the elliptic equations in \eqref{eq:bs-one-lap} for $z = z_k, k =0, 1,
2,\cdots, N,$
are independent each other, no communication is required during the
computation except for the last summation step in the numerical Laplace
inversion. Thus the Laplace transformation method is very well fitted for
parallel computation.
The result in \tabref{tab:bs-basket-s} is generated on $128\times128$ space
grid for $L_1=L_2=300$ with a 15-number of $z$ points using IBM PowerPC97 with 2.2GHz
clock speed.
This table, as can be expected, shows almost ideal speedup 
because of the minimization of communication time. Finally, we attach the
plot of the basket option price at \figref{fig:bs-basket}.
\begin{table}[htbp]
\centering
    \begin{tabular}{|c|c|c|c|c|}\hline
        Number of CPUs & 1 & 3 & 5 & 15  \\\hline
                Time(sec) & 74.93 & 25.25 & 15.31 & 5.671  \\ \hline
                Speedup   & 1.00  & 2.97  & 4.89  & 13.2   \\ \hline
      \end{tabular}
\vspace{.3cm}
\caption{Parallelization speedup in \exref{ex:p-3}}
\label{tab:bs-basket-s}
\end{table}

\begin{figure}
\begin{center}
%  \resizebox{110mm}{!}{\includegraphics{./data/basket-60-64.eps}}
  \resizebox{110mm}{!}{\includegraphics{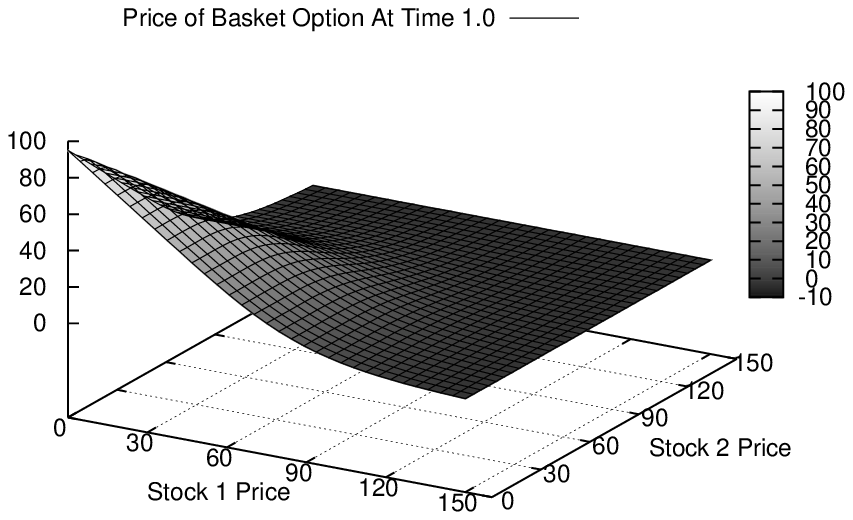}}
  \caption{Basket option price of \exref{ex:p-3}}
  \label{fig:bs-basket}
\end{center}
\end{figure}

\end{document}